\documentclass[runningheads]{llncs}

\usepackage[utf8]{inputenc}
\usepackage[T1]{fontenc}
\usepackage[english]{babel}
\usepackage{amsfonts}
\usepackage{bbm}
\usepackage{amssymb}
\usepackage{latexsym}
\usepackage{array}
\usepackage{subcaption}
\usepackage{color}
\usepackage{pgfplots}
\pgfplotsset{compat=1.14}
\usepackage{stmaryrd}
\usepackage[textwidth=2.5cm, color=yellow]{todonotes}
\usepackage{verbatim}
\usepackage{float}
\usepackage{enumitem}
\usepackage{pdfpages}
\usepackage{eucal}
\usepackage{old-arrows}
\usepackage{tcolorbox}
\usepackage{mathrsfs}
\usepackage{mathtools}  % for left subscripts
\usepackage{comment}
\usepackage{graphicx}
\usepackage{graphics}
\usepackage{xspace}
\usepackage{cite}
\usepackage[vlined, ruled, dotocloa, english, onelanguage, linesnumbered]{algorithm2e}

\newcommand{\cc}[1]{\mathcal{#1}}  % calligraphic
\newcommand{\cb}[1]{\mathbb{#1}}  % doubled (IN, IR, ...)
\newcommand{\csf}[1]{\normalfont{\textsf{#1}}}  % sans serif
\newcommand{\ctt}[1]{{\fontfamily{lmtt}\selectfont #1}}  % typefont

\newcommand{\card}[1]{\vert #1 \vert}  % size of a set
\DeclareMathOperator{\U}{V}  % Universe for variables
\DeclareMathOperator{\cs}{\cc{F}}  % closure system
\DeclareMathOperator{\mt}{\land}  % meet
\DeclareMathOperator{\jn}{\lor}  % join
\DeclareMathOperator{\J}{J}  % join-irreducible set
\DeclareMathOperator{\Hp}{\cc{H}}  % hypergraph 
\DeclareMathOperator{\E}{\cc{E}}  % (hyper)edges
\DeclareMathOperator{\Lt}{\cc{L}}  % lattice
\DeclareMathOperator{\lft}{\normalfont{\csf{left}}}  % left
\DeclareMathOperator{\rht}{\normalfont{\csf{right}}}  % right
\newcommand{\ie}{i.e.,\xspace}

\definecolor{midnight}{RGB}{44, 62, 80}
\definecolor{belize}{RGB}{0, 103, 176}
\definecolor{teal}{RGB}{0, 150, 136}
\definecolor{amethyst}{RGB}{155, 89, 182}
\definecolor{asbestos}{RGB}{127, 140, 141}
\definecolor{clouds}{RGB}{236, 240, 241}
\definecolor{grass}{HTML}{97CE68}
\definecolor{alizarine}{RGB}{231, 76, 60}

\usepackage[colorlinks=true, citecolor=belize]{hyperref}

\begin{document}

\title{Hierarchical decompositions of dihypergraphs \thanks{ 
		The second author is funded by the CNRS, France, ProFan project.}}

%\title{Pouet Contribution Title\thanks{Supported by organization x.}}
%%
%%\titlerunning{Abbreviated paper title}
%% If the paper title is too long for the running head, you can set
%% an abbreviated paper title here
%%
\author{Lhouari Nourine \inst{1} \and
Simon Vilmin\inst{1}}
%%
%\authorrunning{F. Author et al.}
%% First names are abbreviated in the running head.
%% If there are more than two authors, 'et al.' is used.
%%
\institute{LIMOS, Université Clermont Auvergne, Aubière , France \\
\email{simon.vilmin@ext.uca.fr}\\
\email{lhouari.nourine@uca.fr}}

\maketitle              % typeset the header of the contribution
\begin{abstract}
In this paper we are interested in decomposing a dihypergraph $\Hp = (\U,\E)$ into 
simpler dihypergraphs, that can be handled more efficiently.
%a , a less-known representation for closure systems and lattices.
We study the properties of dihypergraphs that can be hierarchically decomposed into 
trivial  dihypergraphs, \ie vertex hypergraph. The hierarchical decomposition is 
represented by a full labelled  binary tree called $\Hp$-tree, in the fashion of 
hierarchical clustering.
We present a polynomial time and space algorithm to achieve such a decomposition by 
producing  its corresponding $\Hp$-tree.
%
%The aim of this decomposition is to recursively partition a dihypergraph without loss 
%of  informations.
However, there are  dihypergraphs that cannot be completely decomposed into trivial 
components. Therefore, we relax this requirement to more indecomposable dihypergraphs 
called H-factors, and discuss applications of this decomposition to closure systems 
and lattices.

\end{abstract}
\section{Introduction}
\label{sec:intro}
%%!TEX root = ./main.tex
%
In this paper we are interested in decomposing  \emph{directed hypergraphs} (\emph{dihypergraphs} for short). 
They are a generalization of directed graphs, as hypergraphs generalize graphs. 
Dihypergraphs are often used to model implication systems in various fields of computer 
science such as databases \cite{ausiello1986minimal, ausiello2017directed}, closure 
systems and lattice theory \cite{bertet2018lattices, wild2017joy}, 
propositional and Horn logic \cite{gallo1993directed, gallo1998max, wild2017joy} for 
instance.

A dihypergraph consists in a finite set of vertices $\U$ and a collection $\E$ of 
(hyper)edges (sometimes called hyperarcs) of the form $(B, h)$ over $\U$, where $B$ is a subset and $h$ a singleton of $\U$. In database theory, $\U$  corresponds to a relation  schema and edges are functional dependencies; whereas in Horn logic an edge is definite Horn clause on the propositional variables set $\U$. In general, an edge $(B, h)$ depicts a causality relation 
between $B$ and $h$, namely, whenever we deal with $B$ we also have to take $h$ into 
consideration. Note however that a more general definition of dihypergraph is given in 
\cite{gallo1993directed, gallo1998max} where the dihypergraphs we use in 
this paper are called $B$-graphs.

We are interested in decomposing a dihypergraph $\Hp=(\U,\E)$ into simpler dihypergraphs, that can be handled more efficiently. The \emph{hierarchical 
decomposition}  (H-decomposition for short) of a dihypergraph considered in this paper, is a {\it recursive 
partitioning} of the 
vertex set  of the dihypergraph into smaller subhypergraphs or clusters, in the fashion of hierarchical clustering 
(see \cite{dasgupta2016cost}). The H-decomposition is a way to represent a dihypergraph 
as a tree while preserving its vertices and edges. 
The notion of a split of a dihypergraph is the principal tool  we will use to achieve the H-decomposition. A split of a dihypergraph $\Hp=(\U,\E)$ is a partitioning of the dihypergraph's vertices into two subset $(\U_1,\U_2)$ such that the edges of $\Hp$ are the disjoint union of the edges of  the induced subhypergraphs  $\Hp[\U_1]$, $\Hp[\U_2]$ and the bipartite dihypergraph $\Hp[\U_1,\U_2]$, \ie for any $e=(B, h) \in 
\E$ intersecting both $\U_1$ and $\U_2$, we  have $B \subseteq \U_1$ and $h \in \U_2$ or 
vice versa. Clearly, there are dihypergraphs that cannot have a split. Our motivation is 
to study  properties of dihypergraphs that can be H-decomposed into trivial 
dihypergraphs, \ie hypergraphs with one vertex. 
The H-decomposition is represented by a full labelled binary trees called H-tree.

An application  for our work arises from the decomposition of closure systems, 
or lattices. The concept of splitting lattices or closure systems is an old 
question and remains an active topic in several areas in 
mathematics and computer science. Among the common ways to split a lattice are the 
subdirect decomposition,  the duplication (or doubling) of convex sets 
\cite{bertet2002doubling,viaud2015reverse},  and other summarised decomposition in 
\cite{ganter1999decompositions, ganter2012formal, gratzer2011lattice, 
kengue2005parallel}. 
The former has 
been early considered by Birkhoff in \cite{birkhoff1944subdirect} where his 
representation theorem 
\textit{``Every algebra is a subdirect product of its subdirectly irreducible homomorphic 
images''} is stated. Jipsen and Rose \cite{jipsen1992varieties} summarize many results 
related to subdirect decomposition and give a list of subdirectly irreducible lattices. 
From the algorithmic point of view, several works can be found in 
\cite{ganter1999decompositions, viaud2015subdirect} where 
closure systems are represented with binary matrices (known as contexts) instead of 
dihypergraphs. 
Database theory community has however provided some 
decomposition schemes for dihypergraphs such as in \cite{demetrovics1992functional, 
spencer1996efficient} or \cite{libkin1993direct}, in view of database normalization. 
Other works on decomposition of dihypergraphs  are  considered  in \cite{berg2003edge, 
gallo1998max, ausiello2017directed, popp2020multilevel}, but these works differ in aims 
and methods from our work.

 In this paper, we present a polynomial time and space algorithm to achieve such a H-decomposition by producing its corresponding H-tree if it exists. However, there are dihypergraphs that cannot be completely decomposed into trivial components. Therefore, we relax this requirement to more indecomposable dihypergraphs called H-factors. This relaxation allows us to extend the H-decomposition of  dihypergraphs to closure systems and lattices. This approach of H-decomposing closure systems permit a deep understanding of the subdirect product via the dihypergraphs representation of closure systems.

%theory.

The paper is structured as follows. In Section \ref{sec:prelim} we recall some 
definitions about dihypergraphs. In Section \ref{sec:structure} we define the 
hierarchical decomposition of dihypergraphs, and its representation by a binary labelled tree. We also give a polynomial time and space algorithm to recognise dihypergraphs having a H-decomposition and produces the tree decomposition. Section 
\ref{sec:clos} extends the H-decomposition to closure systems and provide some properties that can be useful for closure systems classification. 

% In particular we show that Tamari lattices can be represented by a $\Hp$-tree.

\section{Preliminaries}
\label{sec:prelim}
%%!TEX root = ./main.tex

All the objects considered in this paper are finite. For a set  $\U$, we denote by 
$2^{\U}$ its powerset, and for $n \in \cb{N}$, 
we denote by $[n]$ the set $\{1, \dots, n\}$. We also sometimes omit braces for sets, writing $v_1 v_2\dots v_n$ for the set $\{v_1, \dots, v_n\}$.

We mainly refer  to papers \cite{ausiello2017directed, gallo1993directed} for terminology and definitions of dihypergraphs. A \emph{(directed) hypergraph} 
(\emph{dihypergraph} for short) $\Hp$ is a pair $(\U(\Hp), \E(\Hp))$ where $\U(\Hp)$ is its set of 
vertices,  and $\E(\Hp) = \{e_1, \dots, e_n\}$,  $n \in \cb{N}$, its set of \emph{edges}. 
An edge $e\in \E(\Hp)$ is a pair $(B(e), h(e))$, where $B(e) \subseteq \U$  called the \emph{body} of $e$ and $h(e) \in \U \setminus B$ 
 called the \emph{head} of $e$.

 When it is clear from the context, we write $\U$, $\E$ and $(B, h)$ instead of 
 \ $\U(\Hp)$, $\E(\Hp)$ and  $(B(e), h(e))$ respectively. An edge $e=(B, h)$ is  written  as the set $e=B \cup \{h\}$ when no confusion can arise.
Whenever  the body $B$ of an edge is reduced to a single element 
$b$, we shall write $(b, h)$ instead of $(\{b\}, h)$ for clarity. In this case, the edge  $(b, h)$ is called a \emph{unit edge}. 
If all the edges of a dihypergraph are unit, then it is called  a {\it digraph}.
 
 %Given a dihypergraph $\Hp = (\U, \E)$ and two vertices $v,v'\in \U$. A path from $v$ to $v'$ in $\Hp$ is a sequence of $k$ edges $(B_1,h_1)...(B_k,h_k)(B_{k+1},h_{k+1})$ such that $v\in B_1$, $v'=h_k$ and $h_i\in B_{i+1}, i\in[k]$. If $v'\in B_1$ then it is a cycle. A dihypergraph is {\it acyclic} if there is no cycle between any two vertices.

Let $\Hp = (\U, \E)$ be a dihypergraph and $U$ a subset of  $\U$. The subhypergraph $\Hp[U]$ 
\emph{induced} by $U$ is the pair $(U, \E(\Hp[U]))$ where $\E(\Hp[U])$ is the set of 
edges of $\E$ contained in $U$, namely $\E(\Hp[U]) = \{e \in \E \mid e \subseteq U\}$.
 A \emph{bipartite dihypergraph} is a 
dihypergraph in which the ground set can be partitioned  into two parts $(\U_1, \U_2)$ such 
that for any $(B,h) \in \E$, $B \subseteq \U_1$ or $B \subseteq \U_2$. We denote a bipartite dihypergraph by $\Hp[\U_1,\U_2]$.
The \emph{size} of a dihypergraph $\Hp$ is written $\card{\Hp}$ and is given by 
$\card{\Hp} = \card{\U} + \sum_{e \in \E} \card{B(e)} + 1$. The number of edges in $\E$ 
is written $\card{\E}$.

%Given a dihypergraph $\Hp = (\U, \E)$, we define a body-path in $\Hp$ to be  a sequence $v_1,e_1,v_2,...,v_k,e_k,v_{k+1}$ of distinct vertices and edges of $\Hp$ such that: (1) $v_i\in \U, i\in[k+1]$, (b)  $e_i=(B_i,h_i)\in \E, i\in[k]$, and (3) $\{v_i,v_{i+1}\}\subseteq B_i, i\in[k]$. Two vertices $v, v'\in \U$ are said to be connected in $\Hp$ if there exists a body-path from $v$ and $v'$. A dihypergraph $\Hp$ is connected if every pair of vertices $v,v'\in \U$  is connected in $\Hp$. A connected component of a dihypergraph $\Hp$ is a maximal connected  subdihypergraph of $\Hp$.

Let $T = (\U(T), \E(T))$ be a full rooted binary tree and $v \in \U(T)$. We 
denote by $\lft(v)$ its left child and $\rht(v)$ its right one. The subtree induced by $v$
is written $T[v]$, and the leaves of $T[v]$ are given by $\csf{leaves}(v)$.
%As we will use full rooted binary trees as a support for decompositions of 
%dihypergraphs, 
Sometimes, we will write $v \in T$ as a shortcut for $v \in \U(T)$.
We  assume that the ground set $\U(T)$ is disjoint from the ground set of any 
dihypergraph we will deal with.

\section{Hierarchical decomposition of a dihypergraph}
\label{sec:structure}
%%!TEX root = ./main.tex

%\textcolor{red}{
%In this section, we introduce a hierarchical decomposition of a dihypergraph from which 
%it can be fully recovered.
%This decomposition will be represented by a full rooted binary tree labelled with edges 
%and vertices of the dihypergraph: a \emph{$\Hp$-tree}.
%The decomposition consists in recursively applying a partitioning operation we will call 
%a \emph{split}.
%A split of a dihypergraph is a (full) bipartition of its ground set such that the body 
%of 
%any edge is contained in one of the two parts.
%With a split, an edge is either fully contained in a part, or its head and its body lie 
%in distinct parts of the bipartition.
%Hence, a split shows that a dihypergraph can be expressed as two distinct dihypergraphs 
%interacting on each other.
%We show that some dihypergraphs cannot be recursively decomposed using this strategy.
%On the positive side, we give a polynomial time and space algorithm which takes a 
%dihypergraph as an input, and outputs a $\Hp$-tree if it exists.
%First, we need to formally define a split.
%} 

In this section, we introduce a hierarchical decomposition (H-decomposition)  of a dihypergraph, as a recursive partition of the edges into bipartite dihypergraphs,  from which 
it can be fully recovered.
We are interested first  in the class of dihypergraphs that have a hierarchical decomposition. 
Given a dihypergraph $\Hp = (\U, \E)$, we define the partitioning operation called a 
\emph{split} of $\Hp$. 
Then we recursively apply the splitting operation until reaching trivial dihypergraphs. The H-decomposition of a dihypergraph $\Hp$ will be represented by a rooted binary tree, called  
\emph{$\Hp$-tree}. 

We show that not all dihypergraphs can have such a H-decomposition into trivial dihypergraphs, and give a polynomial 
time and space algorithm which takes a dihypergraph as an input, and outputs a $\Hp$-tree 
if it exists. Moreover, we relax the requirement of the H-decomposition into trivial dihypergraphs to H-factors which are body-connected dihypergraphs.

%give a polynomial time and space to recognize if a given dihypergraph has a hierarchical decomposition.
%The proposed algorithm also  produces  the corresponding 
%$\Hp$-tree from $\Hp$ if it exists. 

\subsection{Split operation}

First we define the split operation of a dihypergraph as follows.

\begin{definition}[split] \label{def:split}
Let $\Hp = (\U, \E)$ be a dihypergraph. A  non-trivial bipartition  $(\U_1, \U_2)$  
of the groundset $\U$ is a \emph{split} of $\Hp$, if for any $e = (B, h) \in 
\E$, $B \subseteq \U_1$ or $B \subseteq \U_2$. 
%$B \cap \U_1 = \emptyset$ or $B \subseteq \U_2 = \emptyset$. 
\end{definition}

A split $(\U_1, \U_2)$ induces three subhypergraphs $\Hp[\U_1]$, $\Hp[\U_2]$ and a 
bipartite dihypergraph $\Hp[\U_1,\U_2] = (\U_1, \U_2, \E_{12})$ where $\E_{12} = \{e \in 
\E \mid e \nsubseteq \U_1 \text{ and } e \nsubseteq \U_2\}$. Moreover, the edges of $\Hp[\U_1]$, $\Hp[\U_2]$ and  $\Hp[\U_1,\U_2]$ form a partition of the edges of $\Hp$.  Indeed, no edge is missed by a split.
Intuitively, the split shows that $\Hp$ is fully described by two smaller distincts 
dihypergraphs $\Hp[\U_1]$ and $\Hp[\U_2]$ acting on each other through the bipartite 
dihypergraph $\Hp[\U_1,\U_2]$. 
%
%For a split $\pi$ of $\Hp$, 
%the bipartite subhypergraph induced by the set of edges that cross the 
%bipartition is called a \emph{splitset} of $\Hp$, and is denoted by 
%$\Hp_{\pi} = (\U_1,\U_2,\E_{12})$ where $\E_{12} = \{e \in \E \mid e \nsubseteq \U_1 
%\text{ and } e \nsubseteq \U_2\}$.
%

%The following example shows two possible bipartition of a dihypergraph. One of them is 
%a split, the other is not
%
%The following example shows the difference between a split and the classical notion of 
%cuts in dihypergraph \cite{gallo1998max}.

\begin{example} \label{ex:intro}
Consider the dihypergraph $\Hp = (\U, \E)$ depicted in Figure 
\ref{fig:ex-intro}, with $\U = [7]$ and $\E = \{(12, 3),  \allowbreak (3, 1), (56, 2), 
(23, 7), (45, 6), (5, 7)\}$.
The bipartition illustrated by the full line separates $\U$ in two sets $\{1, 3\}$ and 
$\{2, 4, 5, 6, 7\}$. 
It is not a split since the body of the edge $(12, 3)$ intersects the two parts, and will be missed.
The bipartition corresponding to the dotted line  
$\U_1=\{1, 2, 3\}$ and $\U_2=\{4, 5, 6, 7\}$ is a split, with $\Hp[\U_1]=(\{1, 2, 3\}, 
\{(12, 3),(3, 1)\})$, $\Hp[\U_2]=(  \{1, 2, 3\}, \{ (45, 6), (5, 7)\})$, and 
$\Hp[\U_1,\U_2] = (\{1, 2, 3\} \cup \{4, 5, 6, 7\}, \allowbreak \{(56, 2),(23, 7)\})$. 

\begin{figure}[ht!]
	\centering 
	\includegraphics[scale=0.8]{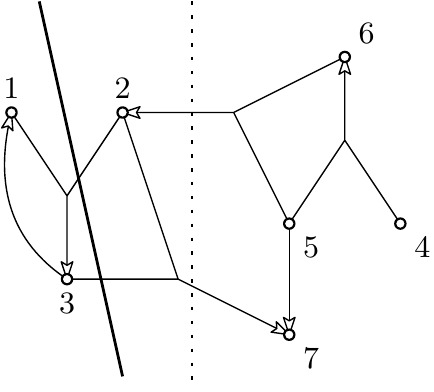}
	\caption{The full line illustrates a bipartition which is not a split, whereas the 
	dotted line corresponds to a split.}
	\label{fig:ex-intro}
\end{figure}
\end{example}

Before giving a characterization of dihypergraphs having a split, we consider some 
special cases.

\begin{itemize}
	\item If the dihypergraph $\Hp$ is a \emph{digraph} or has no edge. Then any 
	bipartition of the ground set is a split. 
	\item However, there are  dihypergraphs that cannot have a bipartition that corresponds to a split.  
	For example, any bipartition of the dihypergraph $\Hp = (\{1, 2, 3\}, \{(12, 3), \allowbreak (13, 
	2)\})$ would miss an edge. 
	For instance, if we consider the bipartition $\U_1=\{1, 2\}$ and $\U_2=\{3\}$, then 
	we capture $(12, 3)$ but not $(13, 2)$, \ie $\Hp[\U_1]=(\{1, 2\}, \emptyset)$, 
	$\Hp[\U_2]=(\{3\}, \emptyset)$, and $\Hp[\U_1,\U_2]=(\{1, 2\}\cup \{3\}, \{(12,3)\})$.
\end{itemize}

In the following, we show that the  dihypergraph's connectivity  is important for the notion of a 
split. 
Given a dihypergraph $\Hp = (\U, \E)$, we define a \emph{body-path} in $\Hp$ to be a 
sequence $v_1,e_1,v_2,...,v_k,e_k,v_{k+1}$ of distinct vertices and edges of $\Hp$ such 
that: (1) $v_i\in \U, i\in[k+1]$, (2)  $e_i=(B_i,h_i)\in \E, i\in[k]$, and (3) 
$\{v_i,v_{i+1}\}\subseteq B_i, i\in[k]$. Two vertices $v, v'\in \U$ are said to be 
\emph{body-connected} in $\Hp$ if there exists a body-path from $v$ and $v'$. A 
dihypergraph $\Hp$ is \emph{body-connected} if every pair of vertices $v,v'\in \U$ is 
body-connected in $\Hp$. A body-connected component of a dihypergraph $\Hp$ is a maximal 
subset of $\U$ where any pair of vertices is body-connected. Figure \ref{fig:notree} shows a body-connected 
dihypergraph.

%
%{\color{red} $ \Hp = ([3], \{(12, 3)\})$. $12$ is body-connected but is not a maximal 
%body-connected subhypergraph of $\Hp$}
%A body-connected component of $\Hp$ is a maximal subset of $\U$ where any pair of 
%vertices is body-connected.
%

%To investigate this question we need to introduce a connectivity condition on bodies of 
%edges. 
%Let $\Hp = (\U, \E)$ be a dihypergraph and let $B_1, \dots, B_n$ be a set 
%of bodies of $\E$. 
%We say that the sequence $B_1, \dots, B_n$ is a \emph{body-path} if 
%for any $1 \leq i < n$ we have $B_i \cap B_{i + 1} \neq \emptyset$. 
%Two vertices  $u, v \in \U$ are \emph{body-connected} if there is a body-path $B_1, 
%\dots, B_n$ such that $u \in B_1$ and $v \in B_n$.
%A subset $C$ of $\U$ is a \emph{body-connected component} of $\Hp$ if for any $u, v \in 
%C$, $u$ and $v$ are body-connected and $C$ is maximal inclusion-wise for this property. 
%%We denote by $\cc{C}(\Hp)$ the set of body-connected components of $\Hp$. 
%If $\Hp$ has a unique body-connected component, we will say that it is 
%\emph{body-connected}.
%The dihypergraph $\Hp = ([3], \{(12, 3), (13, 2)\})$ depicted  in Figure \ref{fig:notree} is body-connected.
%Indeed, $1$ and $2$ are connected by $(12, 3)$, $1$ and $3$ by $(13, 2)$. By 
%transitivity, all the vertices are body-connected.

\begin{figure}[ht]
\centering 
\includegraphics[scale=0.9]{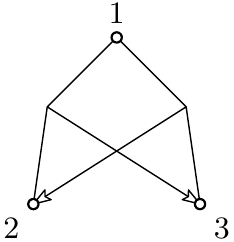}
\caption{A body-connected dihypergraph}
\label{fig:notree}
\end{figure}

First observe that a body reduced to a singleton always satisfies condition of Definition 
\ref{def:split}.  Thus, unit edges of a dihypergraph have no impact on a split.
%
%Let us mention the impact of unit edges in body-connectivity. Indeed, 
%body-connectivity needs to check intersections of bodies of edges. Therefore, a unit edge 
%$(b, h)$ brings no information on the body-connectivity of $b$ since for any other edge 
%$(B', h')$, $B' \cap \{ b\} \neq \emptyset$ implies $b \in B'$. Consequently, the edge 
%$(b, h)$ can be removed from any body-path it belongs to. We are thus lead to the 
%following observation.
%\begin{observation}
%Let $\Hp$ be a dihypergraph. If two vertices are body-connected in $\Hp$, they remain 
%body-connected even if we remove every unit edge from $\Hp$.
%\end{observation}
%
%With the next proposition, we show that body-connectivity is indeed important when we 
%want to check whether a dihypergraph is H-decomposable or not.
Next,  we give a characterization of dihypergraphs that have a split.

\begin{proposition} \label{prop:body-connected}
	A dihypergraph $\Hp$ has a split iff it is not body-connected.
\end{proposition}

\begin{proof} 
Suppose that $\Hp$ has a non trivial split $(\U_1, \U_2)$,  and let $v \in \U_1$ and $v' \in \U_2$.
Assume the existence of a body-path $v=v_1,e_1,v_2,...,v_k,e_k,v_{k+1}=v'$.
Such a body-path exists if there is $i \in [k]$ such that $e_i=(B_i,h_i)$ and $B_i \cap 
\U_1 \neq \emptyset$ and $B_i \cap \U_2 \neq \emptyset$.
But, the edge $e_i = (B_i, h_i)$ cannot satisfy the condition of Definition 
\ref{def:split}. 
Then $v, v'$ are not body-connected and thus $\Hp$ is not body-connected.

Conversely, suppose that $\Hp$ is not body-connected and $C$ be a body-connected 
component of $\Hp$. 
We show that $(C, \U\setminus C)$ is a split. 
Let $e=(B,h)\in \E$. 
Since $C$ is a maximal body-connected component, either $B\cap C=\emptyset$ or 
$(\U\setminus C)\cap B=\emptyset$. 
Hence $(C,\U\setminus C)$ is a split.
\end{proof}

It is important to note that body-connectivity is not inherited. That is, a
subhypergraph induced by a body-connected component may not be body-connected. Consider the dihypergraph 
in Figure \ref{fig:ex-intro} with the split $\U_1=\{1,2,3\}$ and $\U_2=\{4,5,6,7\}$. Then $5$ and $6$ were body-connected in $\Hp$ but not in $\Hp[\U_2]$. Therefore, body-connected components may be 
decomposed in turn.
The main idea of the H-decomposition is to recursively apply the split operation until we 
reach a trivial dihypergraph. 

\subsection{$\Hp$-tree of a dihypergraph}

Based on the split operation, we define the H-decomposition of a dihypergraph. We recursively split a dihypergraph into smaller dihypergraphs until we reach a trivial dihypergraph.
This recursive decomposition can be conveniently represented by a full rooted binary 
tree. An interior node of the tree corresponds to a split $(\U_1,\U_2)$ whose children 
correspond to the H-decomposition of $\Hp[\U_1]$ and $\Hp[\U_2]$; the leaves 
of the tree represent the ground set.
Since the splits $(\U_1,\U_2)$ and $(\U_2,\U_1)$ are the same, the order of the children of an interior node is not important. 

\begin{definition}[$\Hp$-tree] \label{def:is-labelling}
Let $\Hp = (\U, \E)$ be a dihypergraph, $T$ be a full rooted binary tree. Then $(T,\lambda)$ 
is a $\Hp$-tree of $\Hp$ if there exists a 
labelling map  $\lambda \colon T \rightarrow \U \cup 
2^{\E}$ satisfying the following conditions:
\begin{enumerate}[label=(\roman*)]
	\item $\lambda(v) \in \U$ if $v$ is a leaf of $T$,
	\item $\lambda(v) \subseteq \E$ if $v$ is an interior node (possibly 
	$\lambda(v) = \emptyset$),
	\item for any $(B, h) \in \lambda(v)$, elements of $B$ are labels of leaves 
	in the subtree of one child of $v$ and $h$ is the label of a leaf in the subtree of 
	the other child. 
%	then for any $b \in B$, there exists a leaf 
%	$\ell$ in the left subtree (resp. right subtree) of $v$ such that $\lambda(l) = b$ 
%	and there is a leaf $\ell'$ in the right subtree (resp. left subtree) of $v$ such 
%	that $\lambda(\ell') = h$,
%	\item[$L_p$]: if $(B, h) \in \lambda(v)$, then for any $b \in B$, there exists a leaf 
%	$\ell$ in the left subtree (resp. right subtree) of $v$ such that $\lambda(l) = b$ 
%	and there is a leaf $\ell'$ in the right subtree (resp. left subtree) of $v$ such 
%	that $\lambda(\ell') = h$,
	
	\item the set $\{\lambda(v) \mid v \in T\}$ is a full partition of \ $\U \cup 
	\E$ and may contain the emptyset. 
\end{enumerate}
If such labelling exists we call the dihypergraph hierarchically decomposable (H-decomposable for short), and H-indecomposable otherwise.
\end{definition}

%
%
%Let us discuss the conditions we defined. 
%The discussion is followed by some example of $\Hp$-trees. 
%Thanks to conditions \emph{(ii), (iii)}, the label of an interior node $v$ must be the 
%edges of a splitset of the dihypergraph induced by $\csf{leaves}(v)$. 
%Condition \emph{(i)} not only stops the decomposition when it reaches single vertices, 
%but it also guarantees that the resulting $\Hp$-tree has a number of nodes bounded by 
%$\card{\U}$ and hence an overall size bounded by $\card{\Hp}$.
%Finally, condition \emph{(iv)} shows that 
%the decomposition is lossless and that the initial dihypergraph can be fully recovered 
%from the tree by a simple traversal (such as DFS or BFS).
%We now give some example of $\Hp$-trees.
%

Figure \ref{fig:hptree-intro} shows a $\Hp$-tree for the dihypergraph in Figure 
\ref{fig:ex-intro}.

\begin{figure}[ht]
\centering 
\includegraphics[scale=0.8]{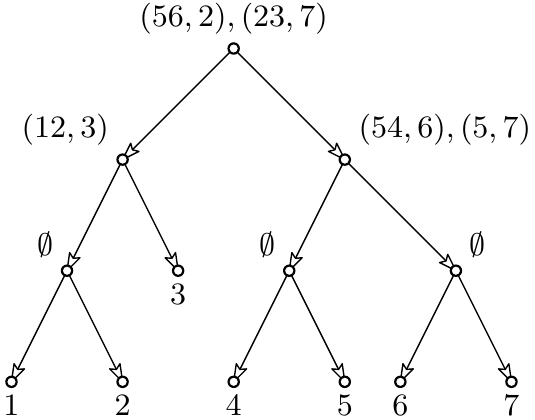}
\caption{A H-decomposition for the dihypergraph  in Figure \ref{ex:intro}}
\label{fig:hptree-intro}
\end{figure}

There are two interesting cases where a H-decomposition of a dihypergraph 
$\Hp$ can be computed easily (see Figure \ref{fig:ex-simple-tree}).
\begin{itemize}
	\item the dihypergraph $\Hp$ has no edges. Here, any full rooted binary tree whose leaves are labelled by a 
	permutation of $\U$ and any interior node by $\emptyset$ is a $\Hp$-tree of $\Hp$. 
	\item $\Hp$ is a \emph{digraph}. The same as for the previous case, except that an edge $(b, h)$ will be in the 
	label of the least common ancestor of the leaves labelled by $b$ and $h$.
\end{itemize}

%We illustrate these constructions in the following example.

%\begin{example} \label{ex:simple-tree}
%Let us consider $\Hp_1 = ([4], \emptyset)$ and $\Hp_2 = ([4], \{(1, 2), (2, 3), (3, 4), 
%(4, 1)\})$. Figure \ref{fig:ex-simple-tree} illustrates the hierarchical decomposition of  $\Hp_1$ and  $\Hp_2$.

\begin{figure}[ht!]
\centering
\begin{subfigure}{0.4\linewidth}
	\centering
	\includegraphics[scale=0.8]{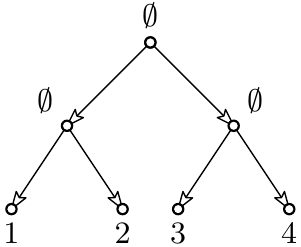}
	%\subcaption{A $\Hp_1$-tree}
\end{subfigure}
\begin{subfigure}{0.4\linewidth}
	\centering 
	\includegraphics[scale=0.8]{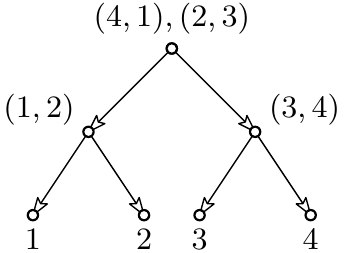}
	%\subcaption{A $\Hp_2$-tree}
\end{subfigure}
\caption{ Hierarchical decompositions for the empty dihypergraphs   $\Hp_1 = ([4], \emptyset)$, and  for the directed graph $\Hp_2 = ([4], \{(1, 2), (2, 3), (3, 4), 
(4, 1)\})$} 
\label{fig:ex-simple-tree}
\end{figure}

%\end{example}

However, there are also some dihypergraphs that cannot be H-decomposed. 

\begin{proposition}\label{prop:H-body-connected}
	If \ $\Hp$ is H-decomposable then it is not body-connected.
\end{proposition}

\begin{proof}
Suppose that $\Hp$ is H-decomposable, and let $(T, \lambda)$ be a $\Hp$-tree with root 
$r$. 
Let $(\U_l, \U_r)$ be the split of $\U$ corresponding to $r$, \ie $\U_l$ corresponds to 
the leaves of the left subtree of $r$ and $\U_r$ to those of the right subtree. Then 
according to Proposition \ref{prop:body-connected}, $\Hp$ is not body-connected.
\end{proof} 

%Figure \ref{} shows a counter example of the  converse of Proposition \ref{prop:H-body-connected}. 
Now, we show that H-decomposability is hereditary, \ie if a dihypergraph $\Hp$ has a 
$\Hp$-tree then any of its subhypergraphs has a H-decomposition too.

\begin{proposition}\label{prop:subset}
Let $\Hp = (\U, \E)$ be a dihypergraph and $U \subseteq \U$. If \ $\Hp$ is 
H-decomposable, so is $\Hp[U]$.
\end{proposition}

\begin{proof}
Let $\Hp = (\U, \E)$ be a dihypergraph, $U \subseteq \U$ and $(T, \lambda)$ a $\Hp$-tree.
We construct a subtree not necessarily induced by $T$ which corresponds to a 
$\Hp[U]$-tree. 
We start from the root $r$ of $T$ and apply the following operation for any interior node 
$v$: if the sets of leaves of the left child and those of the right one intersect 
both $U$, then keep $v$ with label $\lambda(v)= \lambda(v) \cap \E(\Hp[U])$. 
Otherwise, there is a child  of $v$ whose set of leaves do not intersect $U$, in 
this case replace $v$ by the child whose set of leaves intersects $U$. 
The obtained subtree has $U$ as the set of leaves, and the set of labels of 
the internal nodes are exactly  $\E(\Hp[U])$. 	
\end{proof}

%With this proposition at hand, we can now expose the structural theorem which tells us 
%that figuring out whether a dihypergraph $\Hp$ is H-decomposable can be decided 
%recursively, notably through the use of its body-connected components. 

The following 
theorem gives the strategy of the algorithm for recognizing which hypergraphs have a  
H-decomposition.

\begin{theorem} \label{thm:characterization}
Let $\Hp = (\U, \E)$ be a non body-connected dihypergraph and $C$ a body-connected 
component of $\Hp$. Then $\Hp$ is H-decomposable if and only if both a 
$\Hp[C]$ and $\Hp[\U \setminus C]$ are H-decomposable.
\end{theorem}

\begin{proof}
The only if part directly follows from Proposition \ref{prop:subset}.
Let us show the if part.
Let $C$ be a body-connected component of $\Hp$ and let $(T_1, \lambda_1)$ be a 
$\Hp[C]$-tree and $(T_2, \lambda_2)$ a $\Hp[\U \setminus C]$-tree. 
We consider a new tree $(T, \lambda)$ such that $T$ has root $r$ with left subtree $T_1$ 
and right subtree $T_2$.
As for $\lambda$, we put $\lambda(v) = \lambda_1(v)$ if $v \in T_1$, $\lambda(v) = 
\lambda_2(v)$ if $v \in T_2$ and $\lambda(r) = \{e \in \E \mid e \notin \E(\Hp[C]) \cup 
\E(\Hp[\U \setminus C]) \}$. 
In words, $\lambda(r)$ contains any edge which is not fully contained in $C$ or $\U 
\setminus C$.
It is clear that conditions \emph{(i), (ii), (iv)} of Definition \ref{def:is-labelling} 
are fulfilled for $(T, \lambda)$ as they are for 
$(T_1, \lambda_1)$, $(T_2, \lambda_2)$ and $C \cup \U \setminus C = \U$. 
Hence, we have to check \emph{(iii)}. 
Let $e = (B, h)$ be an edge in $\lambda(v)$.
If $B \cap C \neq \emptyset$, then $B \subseteq C$ since $C$ is a body-connected 
component 
of $\Hp$.
As $e$ is not an edge of $\Hp[C]$, it follows that $h \in \U \setminus C$.
Dually, if $B \cap C = \emptyset$, then $h \in C$ since $e$ is not in $\Hp[\U \setminus 
C]$.
Therefore, condition \emph{(iii)} is satisfied and $(T, \lambda)$ is a $\Hp$-tree, 
concluding the 
proof.
\end{proof}

Theorem \ref{thm:characterization} suggests a recursive algorithm which computes a 
$\Hp$-tree for $\Hp$ if it is H-decomposable. 
If $\Hp$ is reduced to a vertex $v$, we simply output a tree which is a leaf with label 
$v$.
Otherwise we compute a body-connected component $C$ of $\Hp$ whenever $\Hp$ is not 
body-connected; we label the corresponding node by the edges of $\Hp[C,\U\setminus C]$, and then we recursively call the algorithm on the 
subhypergraphs $\Hp[C]$ and $\Hp[\U \setminus C]$. 
%It remains to return a new tree with root $r$ such that its left child is a $\Hp[C]$-tree 
%and its right child a $\Hp[\U \setminus C]$-tree.
%Observe that the label of $r$ will be the edges of the splitset induced by the 
%bipartition $(C, \U \setminus C)$.
This  strategy is formalized in Algorithm \ref{alg:build-tree}, whose correctness and 
complexity are studied in Theorem \ref{thm:algo-complex}. 

\begin{algorithm}
	\KwIn{A dihypergraph $\Hp = (\U, \E)$}
	\KwOut{A $\Hp$-tree, if it exists, \csf{FAIL} otherwise}
	
	\BlankLine
	\BlankLine
	
	\If{$\Hp$ has one vertex}{
		%\textcolor{red}{ Donner un nom a ce genre d'arbre : par exemple boolean tree}
		
		%construct a binary tree $T_0$ with root $r_0$ and $\card{\U}$ leaves \;
		%create a labelling $\lambda_0$ s.t. leaves of $T_0$ are bijectively mapped on 
		%$\U$
		%and for any interior node $v \in T_0$, $\lambda_0(v) = \emptyset$ \;
		create a new leaf $r$ with $\lambda(r)$ the unique vertex in $\Hp$\; 
		return $r$ \;
	}\Else{
		compute a body-connected component $C$ of $\Hp$ \;
		\If{$\card{C} =\card{\U}$}{
			stop and return \csf{FAIL} \;
	 
		}
		\Else{
			
			%compute $\Hp[C]$ and $\Hp[\U \setminus C]$ \;
			let $r$ be a new node with $\lambda(r) = \E \setminus (\E(\Hp[C]) \cup 
			\E(\Hp[\U \setminus C]))$ \;
			$\csf{left}(r)= $ \ctt{BuildTree}$(\Hp[C])$ \;
			$\csf{right}(r)  = $ \ctt{BuildTree}$(\Hp[\U \setminus C])$ \;			
			%		let $r$ be a new node with $\lambda(r) = \E \setminus (\E(\Hp[C]) 
			%\cup \E(\Hp[\U 
			%		\setminus C]))$ \;
			%		let $\csf{left}(r) = L$, $\csf{right}(r) = R$ \;
			return $r$ \;
		}
	}
	
	\caption{\ctt{BuildTree}}
	\label{alg:build-tree}
\end{algorithm}

\begin{theorem}\label{thm:algo-complex}
Given a dihypergraph $\Hp = (\U, \E)$, Algorithm {\normalfont \ctt{BuildTree}} computes 
a $\Hp$-tree if it exists and returns {\csf{FAIL}} otherwise in polynomial time and 
space in the size of $\Hp$.
\end{theorem}

%\todo[inline, color=alizarine]{Pas logique $\U \leq \Hp$ si $\E = \emptyset$}

%\todo[inline, color=amethyst]{To be written properly}

\begin{proof}
We first show using induction on the set of vertices $\card{\U}$ that the algorithm 
returns a $\Hp$-tree iff $\Hp$ is H-decomposable. 
Clearly, for dihypergraphs containing only one vertex $x$, the algorithm returns a 
$\Hp$-tree corresponding to a leaf with label $x$. 
Now, assume that the algorithm is correct for dihypergraphs with $\card{\U} < n$, $ 
\in \cb{N}$, and consider a dihypergraph $\Hp$ with $\card{\U} = n$.

Suppose $\Hp$ is H-decomposable. 
Then $\Hp$ is not body-connected by Proposition \ref{prop:body-connected}.
Let $C$ be a body-connected component of $\Hp$.
Inductively, the algorithm is correct for $\Hp[C]$ and $\Hp[\U \setminus C]$ since $1 
\leq \card{C} < n$.
So by Theorem \ref{thm:characterization}, $\Hp[C]$ and $\Hp[\U \setminus C]$ are  
H-decomposable.
By induction the algorithm computes a $\Hp[C]$-tree $(T_1, \lambda_1)$ and a 
$\Hp[\U \setminus C]$-tree $(T_2, \lambda_2)$. 
Therefore, the algorithm returns a tree with root $r$ whose label is $\lambda(r) = \E 
\setminus (\E(\Hp[C]) \cup \E(\Hp[\U \setminus C]))$ and children $T_1$ and $T_2$ which 
satisfies all conditions for $(T, \lambda)$ to be a $\Hp$-tree.
Thus the algorithm computes a $\Hp$-tree for all dihypergraphs that are H-decomposable. 

Now suppose $\Hp$ is not H-decomposable. We have two cases:
\begin{enumerate}
	\item If $\Hp$ is body-connected then the algorithm returns \csf{FAIL} in Line 
	\textbf{7}.
	\item If $\Hp$ is not body-connected, 
	the algorithm chooses a body-connected component $C$ with $1 \leq \card{C} < n$. 
	By Theorem \ref{thm:characterization}, either $\Hp[C]$ or $\Hp[\U \setminus C]$ is 
	H-indecomposable. 
	Thus by induction, the algorithm will return \csf{FAIL} for the input $\Hp[C]$ or 
	$\Hp[\U \setminus C]$ in Lines \textbf{11}-\textbf{12}. 
	Since the algorithm stops, the output of the algorithm is \csf{FAIL}.
\end{enumerate}
Therefore, the algorithm fails whenever the input dihypergraph $\Hp$ is H-indecomposable.
We conclude that the algorithm returns a $\Hp$-tree if and only if the input dihypergraph 
$\Hp$ is H-decomposable.

Now we show that the total time and space complexity of the algorithm are polynomial. 
The space required for the algorithm is bounded by the size of the dihypergraph and 
the size of the $\Hp$-tree.
As the size of the $\Hp$-tree is bounded by $O(\card{\Hp})$, the overall space is bounded 
by $O(\card{\Hp})$.

The time complexity is bounded by the sum of the costs of all nodes (or calls) of 
the search tree. 
The number of calls is bounded by $O(\card{\U})$, the size of the search tree.
The cost of a call is dominated by the computation of a body-connected 
component of the input $\Hp$.
For this, we use union-find data structures  in \cite{tarjan84}, which runs in almost 
linear time, \ie  $O(\card{\Hp} \cdot \alpha(\card{\Hp}, \card{\U}))$ where $\alpha(.,.)$ is the inverse 
of the Ackermann function. 
The almost linear comes from the fact that $\alpha(\card{\U}) \leq 4$ for any practical 
dihypergraph. 
Thus the total time complexity is $O(\card{\U}(  \card{\Hp} \cdot \alpha(\card{\Hp}, \card{\U}))$.
\end{proof}

It is worth noticing, that the obtained $\Hp$-tree by Algorithm \ref{alg:build-tree} depends on the choice of a body-connected component in line {\bf 5}. 
Thus, there are many possible $\Hp$-trees that represent a hierarchical decomposition of a given dihypergraph. 
Then, a natural question arises: are all {\it $\Hp$-trees equivalently interesting?} \ Figure \ref{fig:ex-sev-trees} shows   two possible $\Hp$-trees for the dihypergraph $\Hp 
= (\U, \E)$ with $\U = [8]$ and $\E = \{(12, 3), (23, 4), (34, 5), \allowbreak (56, 7), \allowbreak 
(67, 
8)\}$.

\begin{figure}[ht!]
	\centering
	\begin{subfigure}{0.4\linewidth}
		\centering
		\includegraphics[scale=0.8]{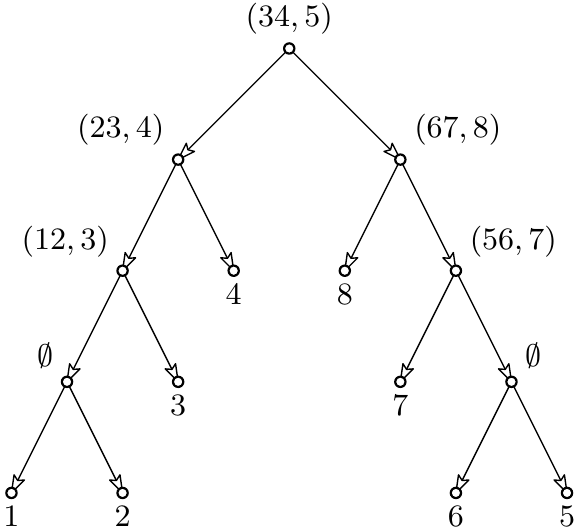}
	\end{subfigure}
	\quad
	\begin{subfigure}{0.4\linewidth}
		\centering
		\includegraphics[scale=0.8]{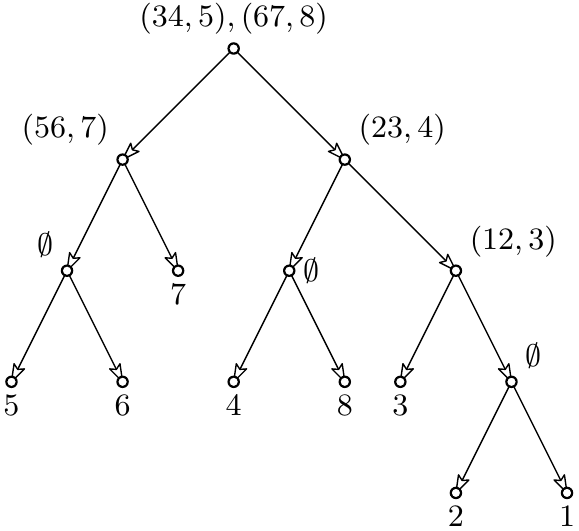}
	\end{subfigure}
	\caption{Two possible $\Hp$-trees of the same dihypergraph}
		\label{fig:ex-sev-trees}
\end{figure}

%\textcolor{red}{A voir si on la garde, si oui donner plus du details
%%
%A possible measure for the quality of a $\Hp$-tree is whether it is balanced or not, as 
%this measure is natural according to binary trees. That is, we 
%try to minimize the difference between the sizes of left and right children of each 
%nodes, and then try to optimize the sum of all nodes. In Figure \ref{fig:ex-sev-trees}, 
%the most balanced $\Hp$-tree is the second one. We are thus lead to consider the 
%following problem
%%
%\paragraph{Problem 1.} Given a dihypergraph $\Hp$, find the $\Hp$-tree which 
%is the most balanced.
%%
%\vspace{1.2em}
%%
%Observe that a naive technique would be to find a split which balances optimally the 
%bipartition at each step, independently of the subhypergraphs it induces. The 
%dihypergraph 
%given in the previous example shows however that this is not a valid strategy. Indeed 
%the 
%first 
%tree begins with splitting the ground set into two equal parts while the second tree 
%first 
%splits $\U$ into non-equal ones. Hence, one would expect the first tree to have a better 
%balancing in the end, but this is not the case. 
%}

\subsection{Extension of the H-decomposition}

As seen before, there are  dihypergraphs that cannot have a split and thus a H-decomposition into trivial hypergraphs.
Such dihypergraphs are body-connected,  and will be called  {\it irreducible H-factors} (H-factors for short) in the rest of the paper.  
Now  we describe a slight modification of Algorithm \ref{alg:build-tree} to obtain a H-decomposition of dihypergraphs into H-factors. Instead of returning \csf{FAIL} in  line {\bf 7} in  Algorithm  {\normalfont \ctt{BuildTree}},
we replace it  by the following:

\vspace{0.8em}

\begin{tabular}{l l}
\textbf{7'} & create a new leaf $r$ with $\lambda(r) = \E$; \\
 & return $r$;

\end{tabular}

\vspace{0.8em}

Figure \ref{fig:M3-factors} illustrates the H-decomposition of a dihypergraph, where the leftmost leaf corresponds to a H-factor which is not trivial.
\begin{figure}[h!]
	\centering 
	\includegraphics[scale=0.8]{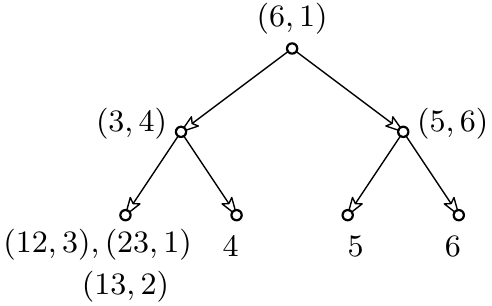}
	\caption{H-decomposition into H-factors}
	\label{fig:M3-factors}
\end{figure}

Now, any dihypergraph has a H-decomposition into H-factors, and then it can be applied to 
any objects encoded by dihypergraphs, as we will show for closure systems in the next 
section.

\section{H-decomposition of a closure system into H-factors}
\label{sec:clos}
%%!TEX root = ./main.tex
Decompositions of closure systems or 
lattices has been widely studied either from the lattice itself \cite{gratzer2011lattice, 
ganter2012formal}, from a context\cite{ganter1999decompositions, 
viaud2015subdirect} or from the database aspect
\cite{libkin1993direct, demetrovics1992functional}.

Decomposition of closure systems is of interest 
for  many applications in Formal Concept Analysis (\cite{ganter2012formal, 
viaud2015subdirect, kengue2005parallel}) such as social networks and datamining.  Closure 
systems are usualy 
represented by a binary matrix, 
also known as context \cite{ganter1999decompositions, 
wild2017joy}.
In this section, we consider closure systems represented by dihypergraphs, see 
\cite{wild2017joy,  ausiello2017directed}, and show that the H-decomposition introduced in the previous section can be applied  to  closure system decomposition.  
%First notice that  Algorithm \ctt{BuildTree} (line {\bf 7}) stops whenever it is  called 
%with a non trivial body-connected dihypergraph as parameter. This means that the 
%dihypergraph contains a subhypergraph which is body-connected. This observation leads us 
%to define a hierarchical decomposition of dihypergraph into body-connected 
%subhypergraphs. 
%

%With a slightly generalized $\Hp$-tree structure
%In the above section we proposed a way to hierarchically decompose a dihypergraph 
%in a $\Hp$-tree. In this section we consider dihypergraphs as a representation for
%closure systems, see \cite{wild2017joy, ausiello2017directed}. First we discuss without 
%proofs the outline of a general algorithm 
%to generate a closure system based on a dihypergraph using the $\Hp$-tree 
%structure. Then we show that particular closure systems, known as Tamari lattices 
%(see \cite{geyer1994tamari, markowsky1992primes}) can be represented by a H-decomposable 
%dihypergraph.

We first recall some definitions for closure systems and lattice theory. The reader can 
refer to \cite{gratzer2011lattice} for a thorough introduction to the topic. A partially 
ordered set $\Lt = (L, \leq)$ is a reflexive, anti-symmetric and transitive binary 
relation 
$\leq$ on a set $L$. For $x, y \in \Lt$, we say that $x$ and $y$ are \emph{comparable} if 
$x \leq y$ or $y \leq x$, and \emph{incomparable otherwise}. An upper bound of $x, y$ is 
an element $u \in \Lt$ such that $x \leq u$, $y \leq u$. If for any upper bound $u' \neq 
u$, $u \leq u'$, then $u$  is the \emph{least upper bound} of $x, y$, written $x \jn y$. 
Lower bounds and the greatest lower bound $x \mt y$ are defined dually. We say that $\Lt$ 
is a \emph{lattice} if for any $x, y \in \Lt$, $x \jn y$ and $x \mt y$ are well defined. 
A {\it  meet-sublattice} $\Lt'$ of $\Lt$ is a subset of elements of $\Lt$ such that for any $x,y\in \Lt'$ $x\wedge y\in \Lt'$. A meet-sublattice $\Lt'$ of $\Lt$ is a {\it sublattice} of $\Lt$ if $x\vee y\in \Lt'$.
Among elements of $\Lt$, we say that $x$ is a \emph{join-irreducible} if for any $y, z 
\in \Lt$, $x = y \jn z$ implies $x = y$ or $x = z$. The set of join-irreducible elements 
of $\Lt$ is denoted by $\J(\Lt)$.

A closure system $\cs$ on a finite set $\U$ is a family of subsets of $\U$ which contains 
$\U$ and is closed under intersection, that is for any $F_1, F_2$ in the family $\cs$, 
$F_1 \cap F_2$ also belongs to $\cs$. A subset $F$ of $\U$ which is in $\cs$ is called a 
\emph{closed set}.  It is well known, that a closure system with partial ordering by set containment
is always  a lattice. Dually, to any lattice $\Lt$ is associated
a closure system on its 
join-irreducible elements. The lattice $\Lt$ is isomorphic to the closure system $\{J_x\mid  x\in \Lt\}$ when ordered by set containment, where  $J_x = \{ j \in \J(\Lt) \mid j \leq x\}$.

The {\it projection} of a closure system $\cc{F}$ over a subset $U \subseteq \U$, named here {\it trace} and noted  $\cs \colon U$,  is the closure system we obtain by intersecting each $F \in \cs$ with $U$, \ie  $\cs \colon U = \{ F \cap U \mid F \in \cs\}$. The trace $\cs \colon U$ is always a sublattice of the lattice $(\cs,\subseteq)$.
The product of two closure systems $\cs_1, \cs_2$ is the pairwise union of their closed sets, that is $\cs_1 \times \cs_2 = \{F_1 \cup F_2 \mid F_1 \in \cs_1, F_2 \in \cs_2 \}$.

%The size of $\cs$ is its number of closed sets.

%It is given by the following transformation: for any $x \in 
%\Lt$, replace $x$ by the set of join-irreducible elements that are below $x$. That is, 
%we replace $x$ by $J_x = \{ j \in \J(\Lt) \mid j \leq x\}$. 

%A closure system can be 
%encoded using a dihypergraph,  where a subset $F$ of $\U$ is closed, if for any 
%edge $(B, h), B\subseteq F$ implies $h\in F$, \ie it is a fixpoint of the 
%forward chaining algorithm on $\Hp$. 

\medskip

%Another well known representation of closure 
%systems is a binary relation (or context) (see \cite{ganter1999decompositions, 
%wild2017joy}).

%\subsection{H-decomposition of a closure system into H-factors}
%Based on the fact that for any dihypergraph is associated a closure system,  we conclude that the  H-decomposition of a dihypergraph induces a H-decomposition of the associated closure system.

First, we recall the {\it forward chaining } method for computing the closure system from its associated dihypergraph. Let $\Hp$ be a dihypergraph and $X\subseteq \U$, we construct a chain of subsets of $\U$ $X=X_0\subset X_1\subset ...\subset X_k=X^{\Hp}$, where  $X_i=X_{i-1}\cup \{h \mid B\subseteq X_{i-1}, (B,h)\in \E\}$ with  $i>0$. The subset $X^{\Hp}$ is called a fixed point or a closed set. Indeed, a subset $F$ of $\U$ is closed, if for any edge $(B, h), B\subseteq F$ implies $h\in F$. The set of all closed sets $\mathcal{F}_{\Hp}=\{X^{\Hp} \mid  X\subseteq \U\}$ is a closure system. Notice, that there are many dihypergraphs that lead to the same closure system.

Naturally, we wish to extend the H-decomposition of a dihypergraph $\Hp$ to a decomposition of the closure system $\mathcal{F}_{\Hp}$, also called H-decomposition. The H-decomposition of the closure system $\mathcal{F}_{\Hp}$ is obtained from the H-decomposition of the dihypergraph $\Hp$, where the label of  a node of its $\Hp$-tree is replaced by  the closure system associated to the dihypergraph induced by  its subtree. The closure systems corresponding to leaves are the irreducible H-factors of the input closure system.

Figure  \ref{fig:lattice-tree} illustrates the H-decomposition of the closure system associated to the H-decomposition of the  dihypergraph  in Figure \ref{fig:M3-factors}.

\begin{figure}[h!]
	\centering 
	\includegraphics[scale=0.8]{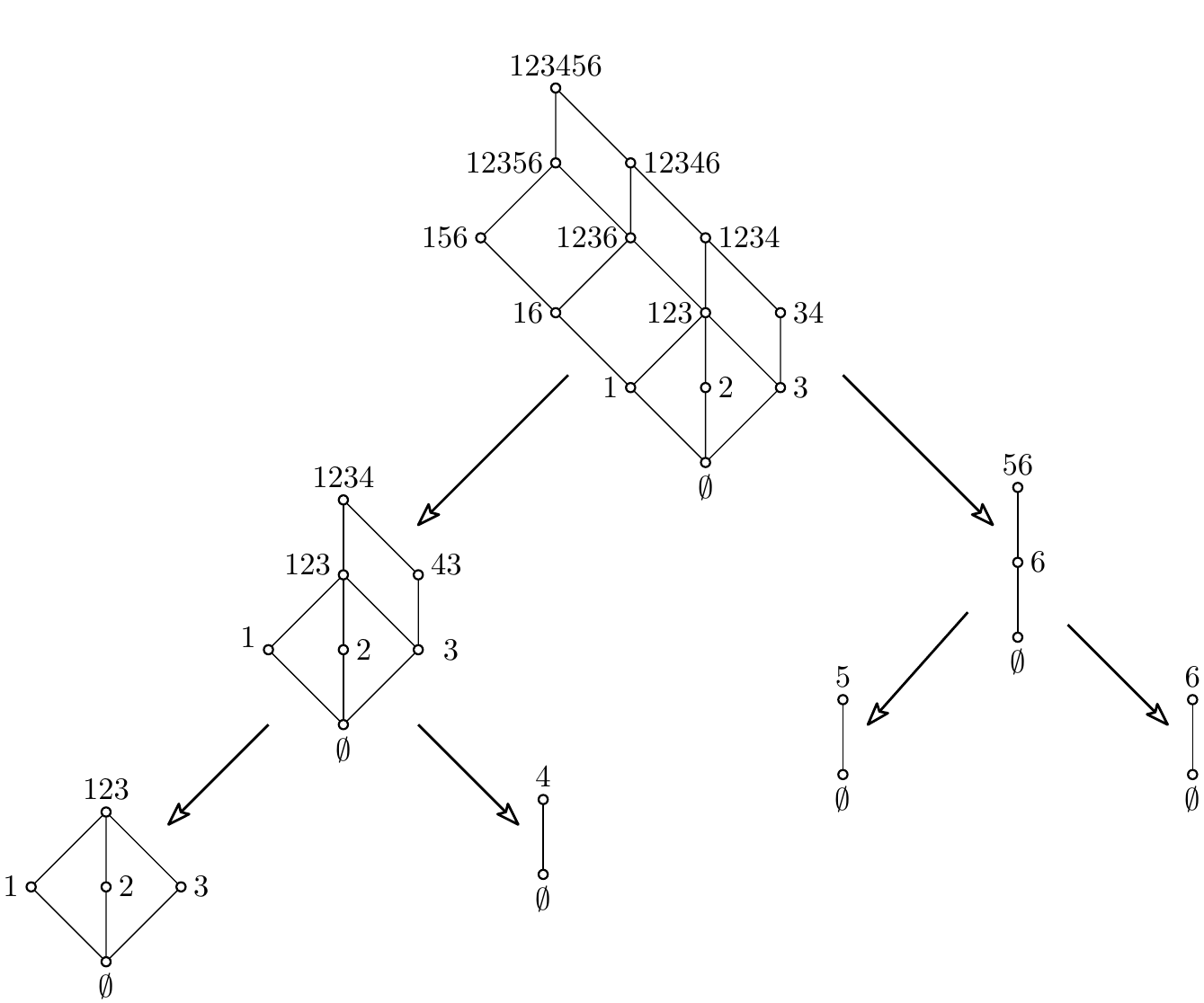}
	\caption{The H-decomposition of the closure system corresponding to the dihypergraph in  Figure 
		\ref{fig:M3-factors}}
	\label{fig:lattice-tree}
\end{figure}

Next, we  study properties of the three closure systems corresponding to the three subhypergraphs induced by a split of the dihypergraph. 

\begin{theorem}  \label{thm:trace}Let  $(\U_1, \U_2)$ be a split of \ $\Hp$,  $\cc{F}_1$ and $\cc{F}_2$ the closure systems corresponding to $\Hp[\U_1]$ and $\Hp[\U_2]$ respectively. Then,
 \begin{enumerate}
 \item If  $F\in \cc{F}_{\Hp}$ then  $F_i=F\cap \U_i\in \cc{F}_i, i=\{1,2\}$. Moreover, $\cc{F}_{\Hp}\subseteq \cc{F}_1\times \cc{F}_2$
 \item If $\Hp[\U_1,\U_2]$ has no edge then $\cc{F}_{\Hp}=\cc{F}_1\times \cc{F}_2$.
 \item If every edge $(B,h)$ of $\Hp[\U_1,\U_2]$, we have $B\subseteq \U_1$    then  $\cc{F}_{\Hp}:\U_1=\cc{F}_1$ and $\cc{F}_{\Hp}:\U_2=\cc{F}_2$. 
  \item If every edge $(B,h)$ of $\Hp[\U_1,\U_2]$, we have $B\subseteq \U_2$    then  $\cc{F}_{\Hp}:\U_2=\cc{F}_2$ and $\cc{F}_{\Hp}:\U_1=\cc{F}_1$. 
  %\item If  $\Hp[\U_1,\U_2]$ is acyclic  then  $\cc{F}_i, i=\{1,2\}$  are traces of \ $\cc{F}_{\Hp}$. 
   % \item If  $\Hp[\U_1,\U_2]$ is general  then $\cc{F}_{\Hp}\subseteq \cc{F}_1\times \cc{F}_2$. 
 \end{enumerate}
\end{theorem}

\begin{proof} Consider a split $(\U_1, \U_2)$  of \ $\Hp$,  $\cc{F}_1$ and $\cc{F}_2$ the closure systems corresponding to $\Hp[\U_1]$ and $\Hp[\U_2]$. 
We will prove \textit{(i)}, \textit{(iii)} and \textit{(ii)}. Item \item{(iv)} is similar to \textit{(iii)}.
\begin{enumerate}
\item[\textit{(i)}] Let $F\in \cc{F}_{\Hp}$ with  $F_i=F\cap \U_i$ and $(B,h)$ an edge of $\Hp[\U_i]$ for some $i\in\{1,2\}$. Suppose $B\subseteq F_i$ and $h\not \in F_i$. Then we also have  $B\subseteq F$ and $h\not \in F$ which contradicts that $F\in \cc{F}_{\Hp}$, since $(B,h)$ is an edge of $\Hp$.
\item[\textit{(iii)}] Without loss of generality, we prove the case for $i=1$. Let $F\in \cc{F}_1$, we show that  $F^{\Hp}$ (the forward chaining applied to $F$ in $\Hp$) satisfies  $F^{\Hp}\cap \U_1= F$. %Clearly, if $\Hp[\U_1,\U_2]$ has no edge then $F^{\Hp}=F$, and it is done. 
Let $(B,h)$ an edge of $\Hp$. We distinguish $3$ cases:
\begin{enumerate}
\item if $B\subseteq \U_2$ then $B\not \subseteq F$. Thus the edge  $(B,h)$ has no effect in the forward chaining.
\item if $B\subseteq \U_1$ and  $h\in \U_1$  then  $B\subseteq F$ implies  $F$ contains $h$ since it is closed in $\Hp[\U_1]$.
\item if $B\subseteq \U_1$ and  $h\in \U_2$  then  $(B,h)$ is an edge of $\Hp[\U_1,\U_2]$. Then there is no edge outside $\Hp[\U_1]$ with head in $\U_1$, and thus the forward chaining cannot add an element from $\U_1$.
\end{enumerate}
So $F$ is in the trace of $\cc{F}_{\Hp}$ over $\U_1$ and $\cs_1 \subseteq \cs \colon \U_1$. 
The reverse inclusion is true by the proof of \textit{(i)}.
As for $\cs_2$, observe that $\cs_2 \subseteq \cs$ since there is no edge $(B, h)$ in $\Hp[\U_1, \U_2]$ such that $B \subseteq \U_2$.
%, since the projection of any closed set $F$ over a set $S$ is closed set of $\Hp[S]$.
%We conclude that 
\item[\textit{(ii)}] Since $\Hp$ has no edge, it satisfies \textit{(iii)}. We deduce that $\cc{F}_{\Hp} \colon \U_1=\cc{F}_1$ and $\cc{F}_{\Hp} \colon \U_2=\cc{F}_2$. Thus $\cc{F}_{\Hp}\subseteq \cc{F}_1\times \cc{F}_2$. For the other inclusion, let $F_1\in \cc{F}_1$ and $F_2\in \cc{F}_2$. We show that $F_1\cup F_2\in \cc{F}_{\Hp}$. Let $(B,h)$ be an edge of $\Hp$ such that   $B\subseteq F_1\cup F_2$. Since $\Hp[\U_1,\U_2]$ has no edge, then $(B,h)$ is an edge of  $\Hp[\U_1]$  or an edge of  $\Hp[\U_2]$. In any case $F_1$ or $F_2$ contains $h$. We conclude that  $F_1\cup F_2\in \cc{F}_{\Hp}$.

%\item[] $(v)$ Let $F\in \cc{F}_{\Hp}$ with  $F_1=F\cap \U_1$ and  $F_2=F\cap \U_2$. 

\end{enumerate}
\end{proof}

According to Theorem \ref{thm:trace} \textit{(i)}, any closure system is a subset of the product of its H-factors closure systems. So  the  idea is to compute in parallel $\cc{F}_1$ and $\cc{F}_2$ for every split $(\U_1, \U_2)$ in the $\Hp$-tree, and then use the dihypergraph  $\Hp[\U_1,\U_2]$ to compute  $\cc{F}_{\Hp}$. But this strategy is expensive, since the size of $\cc{F}_1$ and $\cc{F}_2$ may be exponential in the size of $\cc{F}_{\Hp}$. This is the case, when the subhypergraphs $\Hp[\U_1]$ and  $\Hp[\U_2]$ have no edge, and the edges of  $\Hp[\U_1,\U_2]$ are as follows: $\E_1\cup \E_2 = \{ (vv', u)\in {\U_1}^2\times U_2  \mid v \neq v' \} \cup  \{ (uu', v)\in {\U_2}^2\times \U_1  \mid v \neq v' \}$.
Then,  $\cc{F}_1=2^{\U_1}$ and  $\cc{F}_2=2^{\U_2}$ which are exponential sizes, whereas the closure system $\cc{F}_{\Hp}$  has $|\U_1|\times |\U_2|+ |\U|+2$ elements, namely $\emptyset$, $\U$,  any 
singleton element $v \in \U$, and any pair $vu\in  \U_1\times \U_2$. 
It is worth noticing that this combinatorial explosion cannot happen whenever $\cc{F}_1$ and $\cc{F}_2$ are traces. In this case, the size of the closure system $\cc{F}_{\Hp}$ is at most twice the maximum size of  $\cc{F}_1$ and $\cc{F}_2$.

\begin{figure}[ht!]
	
	\centering
	\begin{subfigure}{0.2\linewidth}
		\centering
		\includegraphics[scale=0.8]{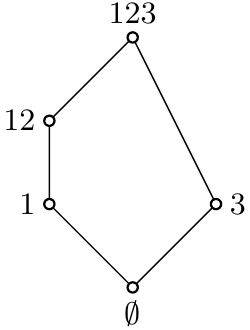}
	\end{subfigure}
	\begin{subfigure}{0.2\linewidth}
		\centering
		\includegraphics[scale=0.8]{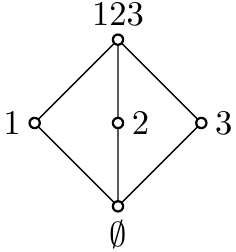}
	\end{subfigure}
	\begin{subfigure}{0.2\linewidth}
		\centering 
		\includegraphics[scale=0.8]{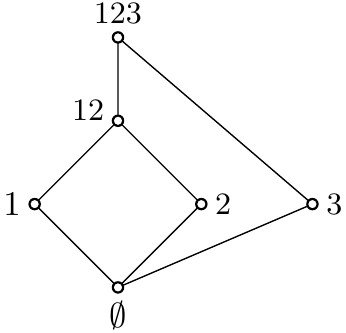}
	\end{subfigure}
	\begin{subfigure}{0.2\linewidth}
		\centering 
		\includegraphics[scale=0.8]{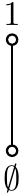}
	\end{subfigure}\\
	(a)~~~~~~~~~~~~~~~~~~~~~~~~~~ (b)~~~~~~~~~~~~~~~~~~~~~~~~~~ (c) 
	~~~~~~~~~~~~~~~~~~~~~~~~~~~~(d)
	\caption{}
	\label{fig:H-indecomposable-factors}
\end{figure}

However, the H-decomposition allows us to go further in the decomposition of closure systems and lattices approaching the most famous Birkhoff's  theorem of the subdirect decomposition which states  "{\it Every algebra is a subdirect product of its subdirectly irreducible homomorphic}".  
The interpretation in our approach is  that "Every closure system is a sublattice of the direct product of irreducible traces". Irreducible traces are closure systems that cannot be obtained as a sublattice of the direct product of its traces.
Consider the closure system $\cc{F}_{\Hp}$ in Figure \ref{fig:H-indecomposable-factors}(a) encoded by the unique dihypergraph  $\Hp=(\{1,2,3\},\{(2,1), (13,2)\})$. It is known that it cannot be obtained as a sublattice of the direct product of  traces.  Clearly $\Hp$ is not body-connected and thus $\U_1=\{1,3\}$ et $\U_2=\{2\}$ is the unique split where $\cc{F}_1=\{\emptyset,1,3,13\}$ and $\cc{F}_2=\{\emptyset, 2\}$ are traces. But $\cc{F}_{\Hp}$ is not a sublattice of  $\cc{F}_1\times \cc{F}_2$, since $(1,3) \in \cc{F}_1\times \cc{F}_2$ the upper bound of $1$ and $3$ is not preserved in $\cc{F}_{\Hp}$.

Figure \ref{fig:H-indecomposable-factors}(b), (c) and (d)  are  subdirectly irreducible and  H-factors too.

We conclude the paper with the following.

\begin{corollary} Every closure system is a meet-sublattice of the direct product of its H-factors. 
\end{corollary}
\begin{proof} This follows from Theorem \ref{thm:trace} \textit{(i)} and the fact that a closure system is closed under intersection.
\end{proof}

%In the future, we are interested in 

%
%
%
%Prerequisite : (w.l.o.g.) for any edge $(A, B)$ of $\Hp[\U_1, \U_2]$, $A = A^{\Hp[\U_1]}$
%and $B = B^{\Hp[\U_2]}$.
%
%\begin{theorem}
%Let $\Hp$ be a dihypergraph and $(\U_1, \U_2)$ a split. Then $\cs_{\Hp} \colon \U_1 =
%\cs_1$ and $\cs_{\Hp} \colon \U_2 = \cs_{2}$ if and only if for any edge $(A, B)$ of
%$\Hp[\U_1, \U_2]$, we have $(B^{\Hp[\U_1, \U_2]} \setminus B)^{\Hp[\U_1]} \subseteq A$
%(or dually if $A \subseteq \U_2$).
%
%\end{theorem}
%\input{tamari}
%\section{Optimization}
%\input{optimization}

%\subsection{Applications to known lattices}
%\input{applications} 
%
%\section{Conclusion}
%\label{sec:concl}
%\input{conclusion}
%
%\nocite{*}

\bibliographystyle{alpha}
\bibliography{Bib-DH}
\end{document}